\documentclass{myifcolog}

\usepackage{doc}

\usepackage{etex}
\usepackage{graphics}
\usepackage{color}
\usepackage{cancel}
\usepackage{stmaryrd}
\usepackage{amsmath}
\usepackage{amsthm}
\usepackage{amssymb}
\usepackage{verbatim,cmll}
\usepackage{setspace}
\usepackage{lscape}
\usepackage{latexsym,relsize}
\usepackage{amscd}
\usepackage{multicol}
\usepackage{bussproofs}
\usepackage[position=top]{subfig}
\usepackage{leqno}
\usepackage{tikz}
\usepackage{authblk}
\usepackage{multicol}
\usetikzlibrary{arrows}
\usetikzlibrary{matrix}
\usetikzlibrary{patterns}
\usetikzlibrary{shapes}
\usepackage{bbm}
\usepackage[bb=boondox]{mathalfa}

\usepackage{lscape}
\usepackage{pdflscape}

\usepackage{color}
\usepackage{colortbl}

\usepackage[greek,english]{babel}

\usepackage{cmll}

\usepackage{comment}

\usepackage{amssymb,amsmath,amsthm}
\usepackage{mathrsfs}
\usepackage{xspace}
\usepackage{bussproofs}
\EnableBpAbbreviations

\usepackage{tikz}

\usepackage{bpextra}

\usepackage{enumerate}

\usepackage{centernot}

\usepackage{mathtools}

\usepackage{hyperref}

\usepackage{cleveref}

\usepackage{relsize}

\usepackage{caption}

\usepackage{multirow}

\usepackage{multicol}

\usepackage{array}
\newcolumntype{C}[1]{>{\centering\arraybackslash}p{#1}}
\newcolumntype{L}[1]{>{\arraybackslash}p{#1}}

\newtheorem{theorem}{Theorem}[section]
\newtheorem{corollary}[theorem]{Corollary}
\newtheorem{lemma}[theorem]{Lemma}
\newtheorem{proposition}[theorem]{Proposition}

\newtheorem{fact}[theorem]{Fact}
\newtheorem{example}[theorem]{Example}

\newtheorem{definition}[theorem]{Definition}

\newtheorem{remark}[theorem]{Remark}

\def\fCenter{{\mbox{$\ \Rightarrow\ $}}}

\newcommand{\fns}{\footnotesize}

\newcommand{\marginnote}[1]{\marginpar{\raggedleft\tiny{#1}}}

\newcommand{\0}{\,\textrm{\bf e}_0}
\newcommand{\1}{\,\textrm{\bf e}_1}
\newcommand{\2}{\,\textrm{\bf e}_2}
\newcommand{\3}{\,\textrm{\bf e}_3}
\newcommand{\4}{\,\textrm{\bf e}_4}
\newcommand{\5}{\,\textrm{\bf e}_5}
\newcommand{\6}{\,\textrm{\bf e}_6}
\newcommand{\7}{\,\textrm{\bf e}_7}

\newcommand{\hh}{\,\textrm{\bf h}}
\newcommand{\ii}{\,\textrm{\bf i}}
\newcommand{\jj}{\,\textrm{\bf j}}
\newcommand{\kk}{\,\textrm{\bf k}}

\newcommand{\mand}{\otimes}

\newcommand{\mrarr}{\,\backslash\,}
\newcommand{\mlarr}{\,\slash\,}

\newcommand{\MAND}{\,\hat{\otimes}\,}
\newcommand{\MRARR}{\,\check{\backslash}\,}
\newcommand{\MLARR}{\,\check{\slash}\,}
\newcommand{\MOR}{\,\check{\oplus}\,}

\newcommand{\MDLARR}{\,\hat{\varobslash}\,}
\newcommand{\MDRARR}{\,\hat{\varoslash}\,}

\newcommand{\aatop}{\ensuremath{\top}\xspace}
\newcommand{\abot}{\ensuremath{\bot}\xspace}
\newcommand{\aand}{\ensuremath{\wedge}\xspace}
\newcommand{\aor}{\ensuremath{\vee}\xspace}

\newcommand{\wdia}{\ensuremath{\Diamond}\xspace}

\newcommand{\bbox}{\ensuremath{{\mkern1.5mu\mkern-2mu\Box\mkern-11.3mu\raisebox{0.8pt}{\rule{1.2ex}{1.2ex}}\mkern2mu}}\xspace}

\newcommand{\WDIA}{\ensuremath{\hat{\rule{0pt}{1.5ex}\Diamond}}\xspace}

\newcommand{\BBOX}{\ensuremath{\mkern1.5mu\check{\mkern-2mu{\Box\mkern-11.3mu\raisebox{0.8pt}{\rule{1.2ex}{1.2ex}}}}\mkern2mu}\xspace}

\title{Vector spaces as Kripke frames}

\titlethanks{}

\addauthor[]{Giuseppe Greco\thanks{The research of the first and third author is supported by a NWO grant under the scope of the project ``A composition calculus for vector-based semantic modelling with a localization for Dutch'' (360-89-070).}}{Utrecht University}
\addauthor[]{Fei Liang\thanks{The research of the second author is supported by the Young Scholars Program of Shandong University (11090089964225).}}{School of Philosophy and Social Development, Shandong University}
\addauthor[]{Michael Moortgat}{Utrecht University}
\addauthor[]{Alessandra Palmigiano\thanks{The research of the fourth and fifth author is supported by the NWO Vidi grant 016.138.314, the NWO Aspasia grant 015.008.054, and a Delft Technology Fellowship awarded to the fourth author in 2013.}}{Vrije Universiteit Amsterdam \\ Department of Mathematics and Applied Mathematics, University of Johannesburg}
\addauthor[]{Apostolos Tzimoulis}{Vrije Universiteit Amsterdam}

\authorrunning{Greco, Liang, Moortgat, Palmigiano, and Tzimoulis}

\titlerunning{Vector spaces as Kripke frames}

\begin{document}
\nopagenumber
\maketitle

\begin{abstract}
In recent years, the compositional distributional approach in computational linguistics has opened
the way for an integration of the \emph{lexical} aspects of meaning into Lambek's type-logical grammar program. This approach is based on the observation that a sound semantics for the associative, commutative and unital Lambek calculus can be based on vector spaces by interpreting fusion as the tensor product of vector spaces. 

In this paper, we  build on this observation and extend it to a `vector space semantics' for the {\em general} Lambek calculus,
based on  {\em algebras over a field} $\mathbb{K}$ (or $\mathbb{K}$-algebras), i.e.~vector spaces endowed with a bilinear binary product.
Such structures are well known in algebraic geometry and algebraic topology, since  Lie algebras and Hopf algebras are important instances of $\mathbb{K}$-algebras.
Applying results and insights from duality and representation theory for the algebraic semantics of nonclassical logics, we regard $\mathbb{K}$-algebras
as `Kripke frames' the complex algebras of which are complete residuated lattices.

This perspective makes it possible to establish a systematic connection between  vector space semantics and the standard Routley-Meyer semantics of (modal) substructural logics. \end{abstract}

\section{Introduction}
The extended versions of the Lambek calculus \cite{lambek1958mathematics,lam61} currently used in computational
syntax and semantics can be considered as multimodal substructural type logics where
residuated families of n-ary fusion operations coexist and interact. Examples are
multimodal TLG with modalities for structural control \cite{Moortgat96},
the displacement calculus of \cite{DBLP:journals/jolli/MorrillVF11} which
combines concatenation and wrapping operations for the intercalation of split strings, or
Hybrid TLG \cite{DBLP:conf/lacl/KubotaL12}, with the non-directional implication of linear logic
on top of Lambek's directional implications. For semantic interpretation, these formalisms
rely on the Curry-Howard correspondence between derivations in a calculus of semantic
types and terms of the lambda calculus that can be seen as recipes for compositional
meaning assembly. This view of compositionality addresses \emph{derivational} semantics
but remains agnostic as to the choice of semantic spaces for \emph{lexical} items.

Compositional distributional semantics \cite{FregeInSpace14, MathematicalFoundationsForACompositionalDistributionalModelOfMeaning,Lambekvs.Lambek:FunctorialVectorSpaceSemanticsAndStringDiagramsForLambekCalculus,TheFrobeniusAnatomyofWordMeaningsI} satisfactorily addresses the lexical aspects of meaning while preserving the compositional view on how word meanings are combined into meanings for larger phrases.
In \cite{Lambekvs.Lambek:FunctorialVectorSpaceSemanticsAndStringDiagramsForLambekCalculus}, the syntax-semantics
interface takes the form of a homomorphism from Lambek's syntactic calculus, or its pregroup variant, to the
compact closed category of finite dimensional vector spaces and linear maps; \cite{LexicalAndDerivationalMeaninginVectorBasedModelsofRelativisation}
have the same target interpretation, but obtain it from the non-associative Lambek calculus extended with a pair of adjoint modal operators allowing for controlled forms of associativity and commutativity in the syntax. The interpretation homomorphism in these approaches typically `forgets' about syntactic fine-structure, sending Lambek's non-commutative,
non-unital syntactic fusion operation to the tensor product of the commutative, associative, unital semantic category, and treating the control modalities
as semantically inert.

In this paper we start exploring a more general interpretation of the Lambek fusion in vector spaces. Our starting point is the notion of {\em algebra over a field} $\mathbb{K}$ (or $\mathbb{K}$-algebra). An algebra over a field $\mathbb{K}$ is a vector space over $\mathbb{K}$ endowed with a bilinear product (cf.~Definition \ref{def:kalgebra}). 
Algebras over a field can be regarded as Kripke (Routley-Meyer) frames in the following way. The vector space structure of a given $\mathbb{K}$-algebra gives rise to a closure operator on the powerset algebra of its underlying vector space (i.e.~the closure operator which associates any set of vectors with the subspace of their linear combinations). The closed sets of this closure operator form a complete non distributive (modular, Arguesian, complemented \cite{J-ModLatt,J-ReprLatt,F-CMLPS}) lattice which interprets the additive connectives ($\wedge, \vee$) of the Lambek calculus (whenever they are considered). The graph of the bilinear product of the $\mathbb{K}$-algebra, seen as a ternary relation, gives rise to a binary fusion operation on the powerset of the vector space in the standard (Routley-Meyer style) way, and moreover the bilinearity of the $\mathbb{K}$-algebra product guarantees that the closure operator mentioned above is a {\em nucleus}. This fact makes it possible to endow the set of subspaces of a $\mathbb{K}$-algebra with a residuated  lattice structure in the standard way (cf.~Section \ref{sec: analysis}). 
This perspective on $\mathbb{K}$-algebras allows us to introduce a more general vector space semantics for the Lambek calculus (expanded with a unary diamond operator and a unary box operator) which we show to be complete (cf.~Section \ref{sec: completeness}), and which  lends itself to be further investigated with the tools of unified correspondence \cite{CoGhPa14,CP-nondist,CPT-Goldblatt} and algebraic proof theory \cite{GMPTZ,GJLPT-LE-logics}.  We start developing some instances of {\em correspondence theory} in this environment, by characterizing the first order conditions on any given (modal) $\mathbb{K}$-algebra corresponding to the validity in its associated (modal) residuated lattice of  several identities involving (the diamond and) the Lambek fusion  such as  commutativity, associativity and unitality. Moreover, using these characterizations, we show that commutativity and associativity fail on the residuated lattice associated with certain well known $\mathbb{K}$-algebras.

\section{Preliminaries}

\subsection{Algebras over a field}

\begin{definition}[\cite{lang1987linear}]
\label{def:vectorspace}
Let $\mathbb{K} = (K, +, \cdot, 0, 1)$ be a field. A {\em vector space} over $\mathbb{K}$ is a tuple $\mathbb{V} = (V, +, \cdot, 0)$\footnote{We overload notation and use the same symbols for sum, product and the constant $0$ both in the field $\mathbb{K}$ and in the vector space $\mathbb{V}$, and rely on the context to disambiguate the reading. Notice that in this axiomatization $-$ is the unary inverse operation and it is considered primitive.} such that 
\begin{itemize}
\item[(V1)] $+: V\times V\to V$ is commutative, associative and with unit $0$;
\item[(V2)] $-: V\to V$ is s.t.~$u+(-u) = 0$ for any $u\in V$;
\item[(V2)] $\cdot: \mathbb{K}\times V\to V$ (called the {\em scalar product}) is an {\em action}, i.e.~$\alpha \cdot (\beta \cdot u)  = (\alpha \cdot \beta)\cdot  u $ for all $\alpha, \beta\in \mathbb{K} $ and every $u\in V$;
\item[(V3)] the scalar product $\cdot$ is {\em bilinear}, i.e.~$\alpha\cdot (u+v) = (\alpha\cdot u) + (\alpha\cdot v)$ and $(\alpha +\beta)\cdot u = (\alpha \cdot u) + (\beta\cdot u)$  for all $\alpha, \beta\in \mathbb{K} $ and all $u, v\in V$;
\item[(V4)] $1\cdot u = u$ for every $u\in V$.
\end{itemize}
\end{definition}

A {\em subspace} $\mathbb{U}$ of a vector space $\mathbb{V}$ as above is uniquely identified by a subset $U\subseteq V$ which is closed under $+, -, \cdot, 0$.

\begin{definition}
\label{def:kalgebra}
An {\em algebra over} $\mathbb{K}$ (or $\mathbb{K}$-{\em algebra}) is a pair $(\mathbb{V}, \star)$ where $\mathbb{V}$ is a vector space $\mathbb{V}$ over $\mathbb{K}$ and  $\star: V\times V\to V$  is {\em bilinear}, i.e.~left- and right-distributive with respect to the vector sum, and {\em compatible} with the scalar product:
\begin{itemize}
\item[(L1$\star$)] $u\star (v+w) = (u\star v) + (u\star w)$ and $(u +v)\star w= (u \star w) + (v \star w)$  for all $u, v, w\in V$;
\item[(L2$\star$)] $(\alpha\cdot u)\star (\beta\cdot v) = (\alpha\beta)\cdot(u\star v)$ for all $\alpha, \beta\in \mathbb{K} $ and all $u, v\in V$.
\end{itemize}
\end{definition}

\begin{definition}
\label{def:subalgebrasV}
A $\mathbb{K}$-algebra $(\mathbb{V}, \star)$ is:
\begin{enumerate}
\item {\em associative} if $\star$ is associative;
\item {\em commutative} if $\star$ is commutative;
\item {\em unital} if $\star$  has a unit $1$;
\item {\em idempotent} if $u =  u\star u$ for every $u\in \mathbb{V}$;
\item {\em monoidal} if $\star$ is associative and unital.
\end{enumerate}
\end{definition}

\begin{example}
Let $\mathbb{R}$ denote the field  of real numbers. A well known example of $\mathbb{R}$-algebra is the algebra  $(\mathbb{H}, \star_H)$ of quaternions \cite{OnQuaternionsandOctonions}, where $\mathbb{H}$ is the 4-dimensional vector space over  $\mathbb{R}$, 
and $\star_H: \mathbb{H}\times \mathbb{H}\to \mathbb{H}$ is 
the  {\em Hamilton product}, defined on the basis elements $\{\1, \ii, \jj, \kk\}$  as indicated in the following table and then extended to $\mathbb{H}\times \mathbb{H}$ by bilinearity as usual. 
Quaternions are the unique associative 4-dimensional $\mathbb{R}$-algebra fixed by $\ii^2 = \jj^2 = \kk^2 = - \1$ and $\ii\jj\kk = -\1$.

\begin{center}
\begin{tabular}{c||c|c|c|c}
$\star_H$ & $\1$ & $\ii$  & $\jj$   & $\kk$ \\
\hline
\hline
$\1$      & $\1$ & $\ii$  & $\jj$   & $\kk$ \\
\hline
$\ii$       & $\ii$ & $-\1$ & \cellcolor{gray!25}$\kk$  & \cellcolor{gray!65}$-\jj$ \\
\hline
$\jj$       & $\jj$ & \cellcolor{gray!25}$-\kk$ & $-\1$ & \cellcolor{gray!25}$\ii$  \\
\hline
$\kk$      & $\kk$ & \cellcolor{gray!65}$\jj$  & \cellcolor{gray!25}$-\ii$  & $-\1$ \\
\end{tabular}
\end{center}
The Hamilton product is monoidal (cf.~Definition \ref{def:subalgebrasV})\footnote{Given our convention, in this case $1$ is an abbreviation for $1\1 + 0 \ii + 0\jj + 0\kk$.} and, notably, not commutative. 
\end{example}

\begin{example}
Another well known example is the $\mathbb{R}$-algebra $(\mathbb{O}, \star_o)$ of octonions \cite{OnQuaternionsandOctonions} where $\mathbb{O}$ is the 8-dimensional $\mathbb{R}$-vector space $\mathbb{O}$, and $\star_O: \mathbb{O}\times \mathbb{O}\to \mathbb{O}$ is defined on the basis elements
 $\0, \1, \2, \3, \4, \5, \6, \7$ 
as indicated in the following table. 

\begin{center}
\begin{tabular}{c||c|c|c|c|c|c|c|c}
$\star_O$ & $\0$ & $\1$ & $\2$ & $\3$ & $\4$ & $\5$ & $\6$ & $\7$ \\
\hline
\hline
$\0$          & \cellcolor{gray!25}$\0$ & \cellcolor{gray!25}$\1$ & \cellcolor{gray!25}$\2$ & \cellcolor{gray!25}$\3$ & \cellcolor{gray!25}$\4$ & \cellcolor{gray!25}$\5$ & \cellcolor{gray!25}$\6$ & \cellcolor{gray!25}$\7$ \\
\hline
$\1$          & \cellcolor{gray!25}$\1$ & \cellcolor{gray!25}$-\0$ & $\3$ & $-\2$ & $\5$ & $-\4$ & $-\7$ & $\6$ \\
\hline
$\2$          & \cellcolor{gray!25}$\2$ & $-\3$ & \cellcolor{gray!25}$-\0$ & $\1$ & $\6$ & $\7$ & $-\4$ & $-\5$ \\
\hline
$\3$          & \cellcolor{gray!25}$\3$ & $\2$ & $-\1$ & \cellcolor{gray!25}$-\0$ & $\7$ & $-\6$ & $\5$ & $-\4$ \\
\hline
$\4$          & \cellcolor{gray!25}$\4$ & $-\5$ & $-\6$ & $-\7$ & \cellcolor{gray!25}$-\0$ & $\1$ & $\2$ & $\3$ \\
\hline
$\5$          & \cellcolor{gray!25}$\5$ & $\4$ & $-\7$ & $\6$ & $-\1$ & \cellcolor{gray!25}$-\0$ & $-\3$ & $\2$ \\
\hline
$\6$          & \cellcolor{gray!25}$\6$ & $\7$ & $\4$ & $-\5$ & $-\2$ & $\3$ & \cellcolor{gray!25}$-\0$ & $-\1$ \\
\hline
$\7$          & \cellcolor{gray!25}$\7$ & $-\6$ & $\5$ & $\4$ & $-\3$ & $-\2$ & $\1$ & \cellcolor{gray!25}$-\0$ \\
\end{tabular}
\end{center}

The product of octonions is unital, but neither commutative nor associative.

\end{example}
\begin{example}
Finally two more examples are the algebras $(\mathbb{M}_n, \star)$, and $(\mathbb{M}_n, \circ_J)$ where $\mathbb{M}_n$ is the vector space of $n\times n$ matrices over $\mathbb{R}$, $\star$ is the usual matrix product and $\circ_J$ is the Jordan product defined as $A\circ_J B =\frac{A\star B + B\star A}{2}$. The usual matrix product is associative but not commutative while the Jordan product is commutative but not associative. 
\end{example}
\subsection{The modal non associative Lambek calculus}
\label{ssec:displaycalc}

The logic of the modal non associative Lambek calculus \textbf{NL}$_{\wdia}$  
can be captured via the \emph{proper} display calculus \textbf{D.NL}$_{\wdia}$ (cf.~\cite{Wa98} where this notion is introduced and \cite{GMPTZ}, which expands on the connection between this calculi and the notion of \emph{analytic structural rules}). Notice that the rules of a Gentzen calculus for this logic are derivable in \textbf{D.NL}. Moreover, the general theory of display calculi guarantees good properties we want to retain, for instance the fact that any display calculus can be expanded with analytic structural rules still preserving a canonical form of cut-elimination. The language of \textbf{D.NL}$_{\wdia}$ is built from the following structural and operational connectives\footnote{Notice that in \cite{Moortgat96} the unary modality $\wdia$ is denoted by the symbols $\diamondsuit$ and $\bbox$ is denoted by the symbol $\Box^\downarrow$.} 
\begin{center}
\begin{tabular}{|r|c|c|c|c|c|}
\hline
\scriptsize{Structural symbols} & $\WDIA$ & $\BBOX$ & $\rule{0pt}{2.60ex}\MAND$ &  $\MLARR$ & $\MRARR$  \\
\hline
\scriptsize{Operational symbols} & $\wdia$  & $\bbox$ & $\rule{0pt}{2.60ex}\mand$ &  $\mlarr$ & $\mrarr$  \\
\hline
\end{tabular}
\end{center}
\noindent The calculus \textbf{D.NL}$_{\wdia}$ manipulates formulas and structures defined by the following recursion, where $p\in\mathsf{AtProp}$: 

\begin{center}
\begin{tabular}{@{}r@{}l@{}}
$\mathsf{Fm} \ni A$ \ & $::=\ p \mid \wdia A \mid \bbox A \mid A \mand A \mid  A \mlarr A \mid  A \mrarr A  $ \\
$\mathsf{Str} \ni X$ \ & $::=\ A \mid \WDIA X \mid \BBOX X \mid X \MAND X \mid  X \MLARR X \mid X \MRARR X $ \\
\end{tabular}
\end{center}

\noindent and consists of the following rules:

\begin{itemize}
\item[] \textbf{Identity and Cut}
\end{itemize}
{
\begin{center}
\begin{tabular}{rl}
\AXC{\phantom{X}}
\LL{\scriptsize $\mathrm{Id}$}
\UI$p \fCenter p$
\DP
 & 
\AX $X \fCenter A$
\AX $A \fCenter Y$
\RL{\scriptsize $\mathrm{Cut}$}
\BI $X \fCenter Y$
\DP \\
\end{tabular}
\end{center}
}

\begin{itemize}
\item[] \textbf{Display postulates}
\end{itemize}
{
\begin{center}
\begin{tabular}{rl}
\AX$Y \fCenter X \MRARR Z$
\RL{\fns $\mand \dashv \backslash$}
\doubleLine
\UI$X \MAND Y \fCenter Z$
\LL{\fns $\mand \dashv \slash$}
\doubleLine
\UI$X \fCenter Z \MLARR Y$
\DP
 & 
\AX$\WDIA X \fCenter Y$
\LL{\fns $\wdia \dashv \bbox$}
\doubleLine
\UI$X \fCenter \BBOX Y$
\DP
 \\
\end{tabular}
\end{center}
}

%

\begin{itemize}
\item[] \textbf{Logical rules}
\end{itemize}
{
\begin{center}
\begin{tabular}{rl}

%
%

\AX$A \MAND B \fCenter X$
\LL{\fns $\mand_L$}
\UI$A \mand B \fCenter X$
\DP
 & 
\AX$X \fCenter A$
\AX$Y \fCenter B$
\RL{\fns $\mand_R$}
\BI$X \MAND Y \fCenter A \mand B$
\DP \\

 & \\

\AX $X \fCenter A$
\AX $B \fCenter Y$
\LL{\fns $\backslash_L$}
\BI$A \mrarr B \fCenter X \MRARR Y$
\DP
 & 
\AX$X \fCenter A \MRARR B$
\RL{\fns $\backslash_R$}
\UI$X \fCenter A \mrarr B$
\DP \\

 & \\

\AX $B \fCenter Y$
\AX $X \fCenter A$
\LL{\fns $\slash_L$}
\BI$B \mlarr A \fCenter Y \MLARR X$
\DP
 & 
\AX $X \fCenter B \MLARR A$
\RL{\fns $\slash_R$}
\UI$X \fCenter B \mlarr A$
\DP \\

 & \\

\AX$\WDIA A \fCenter X$
\LL{\fns $\wdia_L$}
\UI$\wdia A \fCenter X$
\DP
 & 
\AX$X \fCenter A$
\RL{\fns $\wdia_R$}
\UI$\WDIA X \fCenter \wdia A$
\DP
 \\

 & \\

\AX$A \fCenter X$
\LL{\fns $\bbox_L$}
\UI$\bbox A \fCenter \BBOX X$
\DP
 & 
\AX$X \fCenter \BBOX A$
\RL{\fns $\bbox_R$}
\UI$X \fCenter \bbox A$
\DP
 \\

\end{tabular}
\end{center}
}

A modal residuated poset is a structure $P=(P,\leq,\mand,\mrarr,\mlarr,\wdia,\bbox)$ such that $\leq$ is a partial order and for all $x,y,z\in P$ $$x\mand y\leq z \mbox{ iff } x\leq z\mlarr y\mbox{ iff }y\leq x\mrarr z$$ $$\wdia x\leq y\mbox{ iff }x\leq \bbox y.$$

 The calculus \textbf{D.NL}$_{\wdia}$ is sound and complete with respect to modal residuated posets. Indeed every rule given above is clearly sound on these structures, and the Lindenbaum-Tarski algebra of \textbf{D.NL}$_{\wdia}$ is clearly a modal residuated poset (cf.\ Proposition 9 and the discussion before Theorem 4 in \cite{GJLPT-LE-logics}). Furthermore, \textbf{D.NL}$_{\wdia}$ has the finite model property with respect to modal residuated posets (cf.\ \cite[Theorem 49]{GJLPT-LE-logics}).

\paragraph{Analytic Extensions.} As an example of an extension of \textbf{D.NL}$_{\wdia}$ with analytic structural rules,
consider $A \wdia$ and $\wdia C$ below. 


\begin{center}
\begin{tabular}{rl}

\AX$X \MAND (Y \MAND \WDIA Z) \fCenter W$
\LL{\fns $A \wdia$}
\UI$(X \MAND Y) \MAND \WDIA Z \fCenter W$
\DP
 & 
\AX$(X \MAND \WDIA Y) \MAND Z \fCenter W$
\LL{\fns $\wdia C$}
\UI$(X \MAND Z) \MAND \WDIA Y \fCenter W$
\DP
\end{tabular}
\end{center}

These rules replace global forms of associativity or commutativity by controlled forms of restructuring ($A \wdia$) or reordering ($\wdia C$) that have to be explicitly licensed by the
presence of the $\WDIA$ operation. Rules of this form have been used to model long range dependencies: constructions where a question word or relative pronoun
has to provide the semantic content for an unrealized `virtual' element later in the phrase. In the relative clause
$\texttt{key that Alice found \textvisiblespace\ there}$, for instance, the relative pronoun \emph{that} has to make sure that 
the unrealized direct object of $\texttt{found}$ (indicated by \textvisiblespace ) is understood as the key. To make this
possible, typelogical grammars assign a higher-order type to the relative pronoun; the unexpressed object then has the
logical status of a \emph{hypothesis} that can be withdrawn once it has been used to provide the transitive verb with
its direct object.

We illustrate with the following simple
lexicon: $\texttt{key}:n$, $\texttt{that}:(n\backslash n)\slash(s\slash \wdia \bbox np)$, $\texttt{Alice}:np$, $\texttt{found}:(np\backslash s) \slash np,$ $\texttt{there}:(np\backslash s)\backslash (np\backslash s)$. Consider first the judgment \,\mbox{\texttt{key that Alice found} $\ \Rightarrow\ n$}\, 
where the gap \textvisiblespace\ occurs at the right periphery of the clause $\texttt{Alice found \textvisiblespace}$. In the derivations below a dashed inference line abbreviates applications of display postulates or unary logical rules.
The derivation relies on controlled associativity $A \wdia$:

\begin{center}
{\fns
\AXC{$n \fCenter n$}
\AXC{$n \fCenter n$}
\BIC{$n \mrarr n \fCenter n \MRARR n$}

\AXC{$np \fCenter np$}
\AXC{$s \fCenter s$}
\BIC{$np \mrarr s \fCenter np \MRARR s$}
\AXC{$np \fCenter np$}
\dashedLine
\UIC{$\WDIA \bbox np \fCenter np$}
\BIC{$(np \mrarr s) \mlarr np \fCenter (np \MRARR s) \MLARR \WDIA \bbox np$}
\dashedLine
\UIC{$np \ \MAND\ ((np \mrarr s) \mlarr np \ \MAND\ \WDIA \bbox np) \fCenter s$}
\LL{$A\wdia$}
\UIC{$(np \ \MAND\ (np \mrarr s) \mlarr np) \ \MAND\ \WDIA \bbox np \fCenter s$}
\dashedLine

\UIC{$np \ \MAND\ (np \mrarr s) \mlarr np \fCenter s \mlarr \wdia \bbox np$}

\BIC{$(n \mrarr n) \mlarr (s \mlarr \wdia \bbox np) \fCenter (n \MRARR n) \MLARR (np \ \MAND\ (np \mrarr s) \mlarr np)$}

\dashedLine

\UIC{$\underbrace{n}_{\texttt{key}} \,\MAND\, (\underbrace{(n\mrarr n)\mlarr(s \mlarr \wdia \bbox np)}_{\texttt{that}} \,\MAND\, (\underbrace{np}_{\texttt{Alice}} \,\MAND\, \underbrace{(np\mrarr s)\mlarr np)}_{\texttt{found}}) \fCenter n$}
\DP
 }
\end{center}

This example would be derivable also in Lambek's \cite{lambek1958mathematics} Syntactic Calculus,
where associativity is globally available.
But consider what happens when an adverb is added at the end. We then have to prove the judgment \,$\texttt{key that Alice found}$ $\texttt{there} \Rightarrow n$\, where
the gap $\WDIA \bbox np$ occurs in a non-peripheral position. The Syntactic Calculus lacks the expressivity
to derive such examples. With the help of controlled commutativity $\wdia C$
(and $A \wdia$) the derivation goes through:

\begin{center}
{\fns
\AXC{$n \fCenter n$}
\AXC{$n \fCenter n$}
\BIC{$n \mrarr n \fCenter n \mrarr n$}

\AXC{$np \fCenter np$}
\AXC{$s \fCenter s$}
\BIC{$np \mrarr s \fCenter np \MRARR s$}
\AXC{$np \fCenter np$}
\dashedLine
\UIC{$\WDIA \bbox np \fCenter np$}
\BIC{$(np \mrarr s) \mlarr np \fCenter (np \MRARR s) \mlarr \WDIA \bbox np$}
\dashedLine
\UIC{$(np \mrarr s) \mlarr np \ \MAND\ \WDIA \bbox np \fCenter np \mrarr s$}

\AXC{$np \fCenter np$}
\AXC{$s \fCenter s$}
\BIC{$np \mrarr s \fCenter np \MRARR s$}

\BIC{$(np \mrarr s) \mrarr (np \mrarr s) \fCenter ((np \mrarr s) \mlarr np \ \MAND\ \WDIA \bbox np) \MRARR (np \MRARR s)$}
\UIC{$((np \mrarr s) \mlarr np \ \MAND\ \WDIA \bbox np) \ \MAND \ (np \mrarr s) \mrarr (np \mrarr s) \fCenter np \mrarr s$}
\LL{$\wdia C$}
\UIC{$((np \mrarr s) \mlarr np \ \MAND\ (np \mrarr s) \mrarr (np \mrarr s)) \ \MAND \ \WDIA \bbox np \fCenter np \MRARR s$}
\UIC{$np \ \MAND \ (((np \mrarr s) \mlarr np \ \MAND\ (np \mrarr s) \mrarr (np \mrarr s)) \ \MAND \ \WDIA \bbox np) \fCenter s$}
\LL{$A \wdia$}
\UIC{$(np \ \MAND \ ((np \mrarr s) \mlarr np \ \MAND\ (np \mrarr s) \mrarr (np \mrarr s))) \ \MAND \ \WDIA \bbox np \fCenter s$}
\dashedLine
\UIC{$np \ \MAND \ ((np \mrarr s) \mlarr np \ \MAND\ (np \mrarr s) \mrarr (np \mrarr s)) \fCenter s \mlarr \wdia \bbox np$}

\BIC{$(n \mrarr n) \mlarr (s \mlarr \wdia \bbox np) \fCenter 
(n \MRARR n) \MLARR (np \ \MAND \ ((np \mrarr s) \mlarr np \ \MAND\ (np \mrarr s) \mrarr (np \mrarr s)))$}
\dashedLine


\UIC{$\underbrace{n}_{\texttt{key}} \,\MAND\, (\underbrace{(n\mrarr n)\mlarr(s \mlarr \wdia \bbox np)}_{\texttt{that}} \,\MAND\, (\underbrace{np}_{\texttt{Alice}} \,\MAND\, (\underbrace{(np\mrarr s)\mlarr np)}_{\texttt{found}} \ \MAND \ \underbrace{(np\mrarr s) \mrarr (np \mrarr s)}_{\texttt{there}})) \fCenter n$}
\DP
 }
\end{center}

The original modal Lambek calculus is single-type. However, it is possible to generalize this framework to proper \emph{multi-type} display calculi, which retain the fundamental properties while allowing further flexibility. 
Languages with different sorts (also called types in this context) are perfectly admissible and so-called heterogeneous connectives are often considered (e.g.~\cite{Inquisitive,Multitype,BilatticeLogicProperlyDisplayed,GrecoPalmigianoLatticeLogic,linearlogPdisplayed,greco2017multi,NonNormalLogicsSemanticAnalysisAndProofTheory}). In particular, we may admit heterogeneous unary modalities where the source and the target of $\wdia$ and $\bbox$ do not coincide. 

\section{A Kripke-style analysis of algebras over a field}
\label{sec: analysis}

For any $\mathbb{K}$-algebra $(\mathbb{V}, \star)$, the set $\mathcal{S}(\mathbb{V})$ of subspaces of $\mathbb{V}$ is closed under arbitrary intersections, and hence it is a complete sub $\bigcap$-semilattice of $\mathcal{P}(\mathbb{V})$. 
Therefore, by basic order-theoretic facts (cf.~\cite{DaveyPriestley}), $\mathcal{S}(\mathbb{V})$ gives rise to a  closure operator $[-]: \mathcal{P}(\mathbb{V})\to \mathcal{P}(\mathbb{V})$  s.t.~$[X]: = \bigcap\{\mathbb{U}\in \mathcal{S}(\mathbb{V})\mid X\subseteq U\}$ for any $X\in \mathcal{P}(\mathbb{V})$. The elements of $[X]$ can be characterized as {\em linear combinations} of elements in $X$, i.e.~for any $v\in \mathbb{V}$, 
\[v\in [X]\quad \mbox{ iff }\quad v =\Sigma_i \alpha_i \cdot x_i.\]

If $(\mathbb{V}, \star)$ is a $\mathbb{K}$-algebra, let $\otimes: \mathcal{P}(\mathbb{V})\times \mathcal{P}(\mathbb{V})\to \mathcal{P}(\mathbb{V})$ be defined as follows:  \[X\otimes Y: = \{ x\star y\mid  x\in X\mbox{ and } y\in Y \} = \{ z\mid \exists x\exists y(z = x\star y\mbox{ and } x\in X\mbox{ and } y\in Y)\}.\]

\begin{lemma}
\label{lemma:nucleus}
If $(\mathbb{V}, \star)$ is a $\mathbb{K}$-algebra, $[-]: \mathcal{P}(\mathbb{V})\to \mathcal{P}(\mathbb{V})$ is a nucleus on $(\mathcal{P}(\mathbb{V}), \otimes)$, i.e.~for all $X, Y\in \mathcal{P}(\mathbb{V})$,
\[ [X]\otimes [Y]\subseteq [X\otimes Y].\]
\end{lemma}
\begin{proof}
By definition, $[X]\otimes [Y] = \{u\star v\mid u\in [X]\mbox{ and }v\in [Y]\}$.     Let $u\in [X]$ and $v\in [Y]$, and let us show that $u\star v\in [x\star y\mid x\in X\mbox{ and }y\in Y]$. Since $u = \Sigma_{j} \beta_{j} x_{j}$ for $x_{j}\in X$, we can rewrite  $u\star v$ as follows:  $u\star v= (\Sigma_{j} \beta_{j} x_{j})\star v= \Sigma_{j} (( \beta_{j} x_{j})\star v) = \Sigma_{j} \beta_{j}(x_{j}\star v)$; likewise, since  $v= \Sigma_{k} \gamma_{k} y_{k}$ for $y_{k}\in Y$, we can rewrite each $x_{j}\star v$ as $x_{j}\star v = x_{j}\star (\Sigma_{k} \gamma_{k} y_{k}) = \Sigma_{k}(x_{j}\star ( \gamma_{k} y_{k})) = \Sigma_{k}\gamma_{k}(x_{j}\star  y_{k})$. Therefore:
\[u\star v=\Sigma_{j} \beta_{j}(x_{j}\star v) = \Sigma_{j} \beta_{j}(\Sigma_{k}\gamma_{k}(x_{j}\star  y_{k})) = \Sigma_{j} \Sigma_{k}(\beta_{j}\gamma_{k})(x_{j}\star  y_{k}),\]
which is a linear combination of elements of $X\otimes Y$, as required.
\end{proof}
Hence, by the general representation theory of residuated lattices \cite[Lemma 3.33]{galatos2007residuated}, Lemma \ref{lemma:nucleus} implies that the following construction is well defined:\footnote{Notice that in defining the operations, we prefer to use the standard universal and existential modal logic clauses associated to left and right residuals, respectively.}

\begin{definition}
If $(\mathbb{V}, \star)$ is a $\mathbb{K}$-algebra, let $\mathbb{V}^+: = (\mathcal{S}(\mathbb{V}),  \aatop, \abot, \aand, \aor, \otimes, \backslash, /)$ be the complete residuated lattice generated by $(\mathbb{V}, \star)$, 
i.e.~for all $\mathbb{U}, \mathbb{W},\mathbb{Z}\in \mathcal{S}(\mathbb{V})$,
\begin{equation}
\label{eq:residuation}\mathbb{U}\otimes \mathbb{W}\subseteq \mathbb{Z}\quad \mbox{ iff }\quad \mathbb{U}\subseteq \mathbb{Z}/  \mathbb{W}\quad \mbox{ iff }\quad  \mathbb{W}\subseteq \mathbb{U}\backslash\mathbb{Z},\end{equation}
where
\begin{enumerate}
\item $\aatop:= \mathbb{V}$
\item $\abot:= \{0\}$
\item $\mathbb{U} \aor \mathbb{W}: = [z\mid \exists u\exists w (z = u + w\mbox{ and } u\in U\mbox{ and } w\in W)]$
\item $\mathbb{U} \aand \mathbb{Z}: = \mathbb{U} \cap \mathbb{Z}$
\item $\mathbb{U}\otimes \mathbb{W}: = [ z\mid \exists u\exists w(z = u\star w\mbox{ and } u\in U\mbox{ and } w\in W)]$;
\item $\mathbb{Z} / \mathbb{W}: = [u \mid \forall z \forall w(( z = u\star w\mbox{ and } w\in W)\Rightarrow z\in Z)]$;
\item $\mathbb{U} \backslash \mathbb{Z}: = [w \mid \forall u \forall z(( z = u\star w\mbox{ and } u\in U)\Rightarrow z\in Z)]$.

\end{enumerate}
\end{definition}



\begin{lemma}
$[\mathbb{U} \cup \mathbb{W}] = \mathbb{U} \aor \mathbb{W}$.
\end{lemma}
      
\begin{proof}
To show $[\mathbb{U} \cup \mathbb{W}] \subseteq \mathbb{U} \aor \mathbb{W}$,  it is enough to show that $\mathbb{U} \cup \mathbb{W}\subseteq \{z\mid \exists u\exists w (z = u + w\mbox{ and } u\in U\mbox{ and } w\in W)\}$. Let $x \in \mathbb{U} \cup \mathbb{W}$, which implies $x \in U$ or $x \in W$. Without loss of generality, assume that $x \in U$,  the definition of subspace implies that $0 \in \mathbb{W}$. Hence $x \in \{z\mid \exists u\exists w (z = u + w\mbox{ and } u\in U\mbox{ and } w\in W)\}$ by the fact that $x = x + 0$. Conversely, to show $\mathbb{U} \aor \mathbb{W} \subseteq [\mathbb{U} \cup \mathbb{W}]$, let $z \in \mathbb{U} \aor \mathbb{W}$, we need  to show that $z \in [\mathbb{U} \cup \mathbb{W}]$. Since $z = \Sigma_i \alpha_i (u_i + w_i)$ for all $u_i \in U$ and for all $w_i \in W$,  $z = \Sigma_i \alpha_i u_i + \Sigma_i \alpha_i w_i$ for all $u_i \in U$ and for all $w_i \in W$.  Moreover, since for all $u_i \in U$ and for all $w_i \in W$, $\Sigma_i \alpha_i u_i  \in \mathbb{U} \subseteq  \mathbb{U}  \cup \mathbb{W}$ and $\Sigma_i \alpha_i w_i  \in \mathbb{W} \subseteq \mathbb{U} \cup \mathbb{W}$, $\Sigma_i \alpha_i u_i + \Sigma_i \alpha_i w_i \in [\mathbb{U} \cup \mathbb{W}]$ by the definition of $[-]$. Therefore, $z \in [\mathbb{U} \cup \mathbb{W}]$, as required.

\end{proof}

\section{Sahlqvist correspondence for algebras over a field}

\begin{definition}
\label{def:subalgebrasV+}
If   $(\mathbb{V},\star)$ is a $\mathbb{K}$-algebra,   $\mathbb{V}^+  = (\mathcal{S}(\mathbb{V}),\leq , \mand, \mrarr, \mlarr)$ is:
\begin{enumerate}
\item {\em associative} if $\mand$ is associative;
\item {\em commutative} if $\mand$ is commutative;
\item {\em unital} if there exists a 1-dimensional subspace $\mathbb{1}$ such that $\mathbb{U}\mand\mathbb{1} = \mathbb{U} = \mathbb{1}\mand \mathbb{U}$ for all $\mathbb{U}$;
\item {\em contractive} if $\mathbb{U} \subseteq \mathbb{U} \mand \mathbb{U}$ for all $\mathbb{U}$;
\item {\em expansive} if $ \mathbb{U} \mand \mathbb{U}\subseteq \mathbb{U}$ for all $\mathbb{U}$;
\item {\em monoidal} if $\mand$ is associative and unital.
\end{enumerate}
\end{definition}

The following are to be regarded as first-order conditions on $\mathbb{K}$-algebras, seen as `Kripke frames'.

\begin{definition}
\label{def:quasi-subalgebrasV}
A $\mathbb{K}$-algebra $(\mathbb{V}, \star)$ is:
\begin{enumerate}
\item {\em quasi-commutative} if $\forall u, v \in \mathbb{V}\, \exists \alpha \in \mathbb{K}$ s.t.~$u \star v = \alpha (v \star u)$;
\item {\em quasi-associative} if $\forall u, v, w \in \mathbb{V}\, \exists \alpha \in \mathbb{K}$ s.t.~$(u \star v) \star w = \alpha (u \star (v \star w))$ and $\exists \beta \in \mathbb{K}$ s.t.~$u \star (v \star w )= \beta ((u \star v) \star w)$;
\item {\em quasi-unital} if $\exists 1\in \mathbb{V} s.t.~\forall u \in \mathbb{V}\, \exists \alpha, \beta, \gamma, \delta \in \mathbb{K}$ s.t.~$u = \alpha(u \star 1)$ and $u\star 1 = \beta u$ and $u = \gamma (1\star u)$ and $1\star u= \delta u$;
\item {\em quasi-contractive} if $\forall u \in \mathbb{V}\, \exists \alpha \in \mathbb{K}$ s.t.~$u = \alpha(u \star u) $; 
\item {\em quasi-expansive} if $\forall u, v \in \mathbb{V}\, \exists \alpha, \beta \in \mathbb{K}$ s.t.~$ u \star v = \alpha u + \beta v$;
\item {\em quasi-monoidal} if quasi-associative and quasi-unital.
\end{enumerate}
\end{definition}

\begin{remark}
	The notion of quasi-commutativity is strictly weaker than the notion of commutativity in case $\mathbb{K}$ has more than 2 elements. Indeed take the 2-dimensional vector space over $\mathbb{K}$ with base $e_1,e_2$, and define the bilinear map such that $e_1\star e_2= e_1$, $e_2\star e_1=-e_1$ and $e_1\star e_1=0=e_2\star e_2$. Then it is routine to verify that this $\mathbb{K}$-algebra is quasi-commutative but not commutative.
\end{remark}

\noindent In what follows, we sometimes abuse notation and identify a $\mathbb{K}$-algebra $(\mathbb{V}, \star)$ with its underlying vector space $\mathbb{V}$. Making use of definition \ref{def:quasi-subalgebrasV} we can show the following:

\begin{proposition}
\label{prop:V+IffV}
For every $\mathbb{K}$-algebra $\mathbb{V}$, 
\begin{enumerate}
\item $\mathbb{V}^+ $ is commutative iff $\mathbb{V}$ is quasi-commutative;
\item $\mathbb{V}^+ $ is associative iff  $\mathbb{V}$ is quasi-associative;
\item $\mathbb{V}^+ $ is unital iff  $\mathbb{V}$ is quasi-unital;
\item $\mathbb{V}^+ $ is contractive iff  $\mathbb{V}$ is quasi-contractive;
\item $\mathbb{V}^+ $ is expansive iff  $\mathbb{V}$ is quasi-expansive;
\item $\mathbb{V}^+ $ monoidal iff  $\mathbb{V}$ is quasi-monoidal.
\end{enumerate}
\end{proposition}

\begin{proof}
1. For the left-to-right direction, assume that $\mathbb{V}^+$ is commutative and let  $u, v \in \mathbb{V}$. Then $[u] \mand [v] = [v] \mand [v]$. Notice that $[u] \mand [v] = [u \star v] = \{\alpha (u \star v) \,\vert\, \alpha \in \mathbb{K}\}$ and $[v] \mand [u] = [v \star u] = \{\alpha (v \star u) \,\vert\, \alpha \in \mathbb{K}\}$. Hence, $[u] \mand [v] = [v] \mand [v]$ implies that  $u \star v \in [v \star u]$, i.e.~$u \star v = \alpha (v \star u)$ for some $ \alpha \in \mathbb{K}$, as required.  

Conversely, assume that $\mathbb{V}$ is quasi-commutative, and let $\mathbb{U}, \mathbb{W}\in \mathcal{S}({\mathbb{V}})$. To show that $\mathbb{U}\mand \mathbb{W} \subseteq\mathbb{W}\mand\mathbb{U}$, it is enough to show that $u\star w\in \mathbb{W}\mand\mathbb{U}$ for every $u\in \mathbb{U}$ and $w\in \mathbb{W}$. By the assumption that $\mathbb{V}$ is quasi-commutative,  there exists some $\alpha\in \mathbb{K}$ such that  $u\star w = \alpha(w\star u)\in  \mathbb{W}\mand\mathbb{U}$, as required. The argument for $\mathbb{W}\mand\mathbb{U}\subseteq \mathbb{U}\mand \mathbb{W}$ is similar, and omitted.

2. For the left-to-right direction, assume that $\mathbb{V}^+$ is associative and let $u, w, z \in \mathbb{V}$. Then $([u] \mand [w]) \mand [z] = [u] \mand ([w] \mand [z])$. Notice that $([u] \mand [w]) \mand [z] = [u \star w] \mand [z] = [(u \star w) \star z] = \{\alpha ((u \star w) \star z) \,\vert\, \alpha \in \mathbb{K}\}$ and $[u] \mand ([w] \mand [z]) = [u] \mand [w \star z] = [u \star (w \star z)] = \{\alpha (u \star (w \star z)) \,\vert\, \alpha \in \mathbb{K}\}$. Hence, $([u] \mand [w]) \mand [z] = [u] \mand ([w] \mand [z])$ implies that  $(u \star w) \star z = \alpha (u \star (w \star z))$ for some $ \alpha \in \mathbb{K}$ and $u \star (w \star z) = \alpha ((u \star w) \star z)$ for some $ \alpha \in \mathbb{K}$, as required. 

Conversely, assume that $\mathbb{V}$ is quasi-associative, and let $\mathbb{U}, \mathbb{W}, \mathbb{Z}\in \mathcal{S}({\mathbb{V}})$. To show that $(\mathbb{U}\mand \mathbb{W})\mand\mathbb{Z} \subseteq\mathbb{U}\mand(\mathbb{W}\mand \mathbb{Z})$, it is enough to show that $(u\star w)\star z\in \mathbb{U}\mand(\mathbb{W}\mand \mathbb{Z})$ for every $u\in \mathbb{U}$, $w\in \mathbb{W}$ and $z\in \mathbb{Z}$. Since $\mathbb{V}$ is quasi-associative,  there exists some $\alpha\in \mathbb{K}$ such that  $(u\star w)\star z = \alpha(u\star (w\star z))\in \mathbb{U}\mand (\mathbb{W}\mand\mathbb{Z})$, as required. The argument for $\mathbb{U}\mand(\mathbb{W}\mand \mathbb{Z})\subseteq (\mathbb{U}\mand \mathbb{W})\mand\mathbb{Z}$ is similar, and omitted.

3. For the left-to-right direction, assume that $\mathbb{V}^+$ is unital and let  $1\in \mathbb{V}$ such that $\mathbb{1} =[1]$. Then $ [u] = [u] \mand \mathbb{1} = [u\star 1]$ for any $u \in \mathbb{V}$.  Hence,   $u  =\alpha(u\star 1)$ and $u\star 1 = \beta u$,  for some $ \alpha, \beta \in \mathbb{K}$, as required. Analogously, from $ [u] = \mathbb{1} \mand [u]$ one shows that $u  =\gamma(1\star u)$ and $1\star u = \delta u$ for some $ \gamma, \delta \in \mathbb{K}$.

Conversely, assume that $\mathbb{V}$ is quasi-unital, and let $\mathbb{U}\in \mathcal{S}({\mathbb{V}})$. To show that $\mathbb{U}\mand \mathbb{1} \subseteq\mathbb{U}$, it is enough to show that $u\star 1\in \mathbb{U}$ for every $u\in \mathbb{U}$. By assumption,  there exists some $\alpha\in \mathbb{K}$ such that  $u\star 1 = \alpha u\in  \mathbb{U}$, as required. The remaining inclusions are proven with similar arguments which are omitted.

4. For the left-to-right direction, assume that $\mathbb{V}^+$ is contractive and let  $u \in \mathbb{V}$. Then $ [u] \subseteq [u] \mand [u] = [u\star u]$.  Hence,   $u  =\alpha(u\star u)$   for some $ \alpha\in \mathbb{K}$, as required. 

Conversely, assume that $\mathbb{V}$ is quasi-contractive, and let $\mathbb{U}\in \mathcal{S}({\mathbb{V}})$. To show that $\mathbb{U}\subseteq\mathbb{U}\mand \mathbb{U} $, it is enough to show that $u \in \mathbb{U}\mand \mathbb{U}$ for every $u\in \mathbb{U}$. By assumption,  there exists some $\alpha\in \mathbb{K}$ such that  $u = \alpha (u\star u)\in  \mathbb{U}\mand \mathbb{U}$, as required. 

5. For the left-to-right direction, assume that $\mathbb{V}^+$ is expansive and let  $u, v \in \mathbb{V}$. Then, letting $[u, v]$ denote the subspace generated by $u$ and $v$, we have $[u, v]  \mand [u, v] \subseteq [u, v] $, and since    $u \star v \in [u, v]  \mand [u, v] $ we conclude $u\star v \in [u, v]$, i.e.~ $u\star v= \alpha u + \beta  v$   for some $ \alpha, \beta\in \mathbb{K}$, as required. 

Conversely, assume that $\mathbb{V}$ is quasi-expansive, and let $\mathbb{U}\in \mathcal{S}({\mathbb{V}})$. To show that $\mathbb{U}\mand \mathbb{U} \subseteq \mathbb{U}$, it is enough to show that $u\star v\in \mathbb{U}$ for every $u, v \in \mathbb{U}$. By assumption, there exist some $\alpha, \beta\in \mathbb{K}$ such that $u\star v = \alpha u+\beta v\in \mathbb{U}$, as required. 

6. Immediately follows from 2. and 3.
\end{proof}

\subsection{Examples}

\begin{fact}
\label{fact:QuaternionsAreNotquasiCommutative}
The algebra of quaternions $\mathbb{H}$ is not quasi-commutative.
\end{fact}

\begin{proof}
Let $u = \ii+\jj$ and $v = \jj$, then $u \star_H v = \kk- 1$ and $v \star_H u = -\kk-1$. By contradiction, let us assume that $\star_H$ is quasi-commutative, then there exists a real number $\alpha$ s.t.~$\kk -1 = \alpha (-\kk-1) = \alpha (-\kk) - \alpha$. It follows that $\alpha = 1$ and $a=-1$ contradicting the assumption that $\star_H$ is quasi-commutative. 
\end{proof}

\begin{corollary}
$\mathbb{H}^+$ is not commutative.
\end{corollary}
\begin{proof}
Immediate by Fact \ref{fact:QuaternionsAreNotquasiCommutative} and Proposition \ref{prop:V+IffV}.
\end{proof}

\begin{fact}
\label{fact:OctonionsAreNotquasiAssociative}
The algebra $\mathbb{O}$ of octonions  is not quasi-associative.
\end{fact}

\begin{proof}
Let $u = v = w = 1\0 + 2\1 + 3\2 + 5\3 + 7\4 + 8\5 + 11\6 + 12\7$, then $w \star_O u = u \star_O v = -415 \0 + 4 \1 + 6 \2 + 10 \3 + 14 \4 + 96 \5 + 22 \6 + 24 \7$. In order to show that $w \star_O (u \star_O v) \neq (w \star_O u) \star_O v$ is enough to check the first two coordinates: $w \star_O (u \star_O v) = -1887 \0 -266 \1 \ldots \neq -1887 \0 -1386 \1 \ldots = (w \star_O u) \star_O v$. By contradiction, let us assume that $\star_O$ is quasi-associative, then there exists a real number $\alpha$ s.t. $w \star_O (u \star_O v) = \alpha ((w \star_O u) \star_O v)$. It follows that $-1887 \0 = \alpha (-1887)$ and $-266 = \alpha (-1386)$. We observe that $-1887 \0 = \alpha (-1887)$ holds only for $\alpha = 1$, but then $-266 = \alpha (-1386)$ does not hold contradicting the assumption that $\star_O$ is quasi-associative.
\end{proof}

\begin{corollary}
$\mathbb{O}^+$ is not associative.
\end{corollary}
\begin{proof}
Immediate by Fact \ref{fact:OctonionsAreNotquasiAssociative} and Proposition \ref{prop:V+IffV}.
\end{proof}

\section{Modal algebras over a field}

\begin{definition}
\label{definition:R-linear}
A {\em modal} $\mathbb{K}$-{\em algebra} is a triple $(\mathbb{V}, \star, R)$ such that $(\mathbb{V}, \star)$  is a $\mathbb{K}$-algebra and $R \subseteq V \times V$ is 
{\em compatible} with the scalar product, and it preserves the zero-vector:
\begin{itemize}
\item[(L1R)] $v R u \with z R w \ \Rightarrow \ \forall \gamma\delta\, \exists \alpha \beta \ \ (\gamma v + \delta z) R (\alpha u + \beta w)$;
\item[(L2R)] $ t R (\alpha u + \beta v)  \ \Rightarrow \ \exists \lambda \mu\, \exists z w \ \ z R u \,\with\, w R v \,\with\, \lambda z + \mu w = t$.
\item[(L3R)] $x R 0 \ \Leftrightarrow \ x = 0$.
\end{itemize}
\end{definition}


If $(\mathbb{V}, \star, R)$ is a modal $\mathbb{K}$-algebra, let $\wdia: \mathcal{P}(\mathbb{V}) \to \mathcal{P}(\mathbb{V})$ be defined as follows:  \[\wdia X: = R^{-1}[X]  = \{ v \mid \exists u (vRu \mbox{ and } u\in X) \}.\]
\begin{lemma}
\label{lemma:diamond-nucleus}
If $(\mathbb{V}, R)$ is a modal $\mathbb{K}$-algebra, $[-]: \mathcal{P}(\mathbb{V})\to \mathcal{P}(\mathbb{V})$ is a $\wdia$-nucleus on $(\mathcal{P}(\mathbb{V}), \wdia)$, i.e.~for all $X \in \mathcal{P}(\mathbb{V})$,
\[ \wdia [X] \subseteq [\wdia X].\]
\end{lemma}
\begin{proof}
By definition, $\wdia [X] = \bigcup\{R^{-1}[u] \mid u \in [X]\}$. Let $u \in [X]$, assume that $v R u$ and let us show that $v \in [\wdia X]$. Since $u = \Sigma_{j} \beta_{j} x_{j}$ for $x_{j}\in X$, by $L2R$, $\forall j \exists \lambda_j\, \exists v_j \ v_j R x_j \, \with \, \Sigma_j \lambda_j v_j = v$. So $v \in [\wdia X]$. If $X = \varnothing$, then $\wdia [\varnothing] = \wdia \{0\} = R^{-1}[0]$. By $L3R$, $R^{-1}[0] = \{0\} \subseteq [\wdia X]$. 
\end{proof}
Hence, by the generalization of the representation theory of residuated lattices \cite{B-Int-FEP,B-InvNL}, Lemma \ref{lemma:diamond-nucleus} implies that the following construction is well defined: 

\begin{definition}
If $(\mathbb{V}, \star, R)$ is a modal $\mathbb{K}$-algebra, let $\mathbb{V}^+: = (\mathcal{S}(\mathbb{V}), \leq, \mand, \mrarr, \mlarr, \wdia, \bbox)$ be the complete modal residuated lattice generated by $(\mathbb{V}, \star, R)$, 
i.e.~for all $\mathbb{U}, \mathbb{W} \in \mathcal{S}(\mathbb{V})$,
\begin{equation}
\label{eq:adjunction}\wdia \mathbb{U} \subseteq \mathbb{W} \quad \mbox{ iff }\quad \mathbb{U} \subseteq \bbox \mathbb{W},\end{equation}
where
\begin{enumerate}
\item $\wdia \mathbb{U}: = [v \mid \exists u \,(v R u \mbox{ and } u \in U)]$;
\item $\bbox \mathbb{W}: = [u \mid \forall v \,(v R u \ \Rightarrow \ v \in W)]$.
\end{enumerate}
\end{definition}

\begin{remark}
Notice that every linear map $f:\mathbb{V}\to\mathbb{V}$ satisfies the conditions of Definition \ref{definition:R-linear}, and hence functional modal $\mathbb{K}$-algebras $(\mathbb{V}, \star, f)$ can be defined analogously to definition \ref{definition:R-linear} and their associated algebras will be complete modal residuated lattices such that $\wdia $ $f[-]\dashv f^{-1}[-]$ in $\mathcal{S}(\mathbb{V})$. However, if we make use a linear function $f$ (instead of a relation $R$) to define modal $\mathbb{K}$-algebras, then we are not able to show completeness for the full fragment of \textbf{D.NL}$_{\wdia}$. 
\end{remark}

\subsection{Axiomatic extensions of a modal algebra over $\mathbb{K}$}

In order to capture controlled forms of associativity/commutativity, we want to consider
axiomatic extensions of the modal algebras introduced in the previous section. Below, we consider right-associativity and left-commutativity. 

\begin{definition}
\label{def:submodalalgebrasV+}
If $(\mathbb{V}, \star, f)$ is a modal $\mathbb{K}$-algebra, $\mathbb{V}^+: = (\mathcal{S}(\mathbb{V}), \leq, \mand, \mrarr, \mlarr, \wdia, \bbox)$ is:
\begin{enumerate}
\item {\em right-associative} if $(\mathbb{U} \mand \mathbb{W}) \mand \wdia \mathbb{V} \subseteq \mathbb{U} \mand (\mathbb{W} \mand \wdia \mathbb{V})$;
\item {\em left-commutative} if $(\mathbb{U} \mand \mathbb{V}) \mand \wdia \mathbb{W} \subseteq (\mathbb{U} \mand \wdia\mathbb{W}) \mand \mathbb{V}$.
\end{enumerate}
\end{definition}

\begin{definition}
\label{def:submodalalgebrasV}
A modal $\mathbb{K}$-algebra $(\mathbb{V}, \star, R)$ is:
\begin{enumerate}
\item {\em quasi right-associative} if for $u,w,z,v\in\mathbb{V}$ such that $ vRz$, there exists $\alpha,\beta\in\mathbb{K}$ and $v'$ such that $v'R\beta z$ and $(u\star w)\star v=\alpha(u\star (w\star v'))$;
\item {\em quasi left-commutative} if for all $u,w,z,v\in\mathbb{V}$ such that $vRz$ there exists $\alpha,\beta\in\mathbb{K}$ and $v'$ such that $v'R\beta z$ and $(u\star w)\star v=\alpha ((u\star v')\star w)$.
\end{enumerate}
\end{definition}

\begin{proposition}\label{prop: asso+commu+R}
	For every modal $\mathbb{K}$-algebra $(\mathbb{V}, \star, R)$:
	\begin{enumerate}
	\item $\mathbb{V}^+$ is right associative if and only if $(\mathbb{V}, \star, R)$ is quasi right-associative;
	\item $\mathbb{V}^+$ is left commutative if and only if $(\mathbb{V}, \star, R)$ is quasi left-commutative.
	\end{enumerate}
\end{proposition}
\begin{proof}
	 1. For the left to right direction let $u,w,z,v$ such that $vR z$. By the assumption $([u]\mand[w])\mand R^{-1}[[z]]\subseteq [u]\mand([w]\mand R^{-1}[[z]])$. Since $(u\star w)\star v\in ([u]\mand[w])\mand R^{-1}[[z]]$ it follows that $(u\star w)\star v\in [u]\mand([w]\mand R^{-1}[[z]])$, i.e. there exist $
	 \alpha,\beta\in\mathbb{K}$ and $v'\in\mathbb{V}$ with $v'R\beta z$ such that $(u\star w)\star v=\alpha(u\star(w\star v'))$.
	 
	 For right to left direction let $q\in (\mathbb{U} \mand \mathbb{W}) \mand \wdia \mathbb{Z}$, i.e.\ there exists $u\in\mathbb{U},w\in\mathbb{W}$ and $v\in\wdia\mathbb{Z}$ such that $q=(u\star w)\star v$. Since $v\in \wdia\mathbb{Z}$ there exists $z\in\mathbb{Z}$ such that $vR z$. Then by assumption there exist $\alpha,\beta\in\mathbb{K}$ and $v'\in\mathbb{V}$ such that $v'R\beta z$ and $q=\alpha(u\star (w\star v'))$. It holds that $v'\in \wdia\mathbb{Z}$ since $\beta z\in\mathbb{Z}$, and hence $q\in \mathbb{U} \mand(\mathbb{W} \mand \wdia \mathbb{Z})$.  
	 
	 2. For the left to right direction let $u,w,z,v$ such that $vR z$. By the assumption $([u]\mand [w])\mand R^{-1}[[z]]\subseteq ([u]\mand R^{-1}[[z]])\mand [w]$. Since $(u\star w)\star v\in ([u]\mand [w])\mand R^{-1}[[z]]$ it follows that $(u\star w)\star v\in ([u]\mand R^{-1}[[z]])\mand [w]$, i.e. there exist $
	 \alpha,\beta\in\mathbb{K}$ and $v'\in\mathbb{V}$ with $v'R\beta z$ such that $(u\star w)\star v=\alpha((u\star v')\star w)$.
	 
	 For right to left direction let $q\in(\mathbb{U} \mand \mathbb{W}) \mand \wdia\mathbb{Z}$, i.e.\ there exist $u\in\mathbb{U},w\in\mathbb{W}$ and $v\in\wdia\mathbb{Z}$ such that $q= (u\star w)\star v$. Since $v\in \wdia\mathbb{Z}$ there exists $z\in\mathbb{Z}$ such that $vR z$. Then by assumption there exists $\alpha,\beta\in\mathbb{K}$ and $v'\in\mathbb{V}$ such that $v'R\beta z$ and $q=\alpha((u\star v')\star w)$. It holds that $v'\in \wdia\mathbb{Z}$ since $\beta z\in\mathbb{Z}$, and hence $q\in (\mathbb{U} \mand \wdia \mathbb{Z}) \mand \mathbb{W}$.  
\end{proof}	 
	 \begin{remark}
	 	Notice that in case $R$ is a linear function, the inequalities above imply equality. Indeed, e.g.~in the case of right-associativity, if $zR v$, and $\beta zR v'$ then $v'=\beta v$. Therefore, it immediately follows that $(u\star w)\star v=\alpha(u\star (w\star v))$, and hence $\frac{1}{\alpha}((u\star w)\star v) = u\star( w\star v)$, and hence $\mathbb{U}\mand (\mathbb{W}\mand \wdia \mathbb{Z})\subseteq (\mathbb{U}\mand\mathbb{W})\mand \wdia \mathbb{Z}$.
		\end{remark}

\section{Completeness}\label{sec: completeness} 
The aim of this section is to show the completeness of the logic \textbf{D.NL}$_{\wdia}$ with respect to modal $\mathbb{K}$-algebras of finite dimension (cf.\ Theorem \ref{th: wcomp}).

Given a modal $\mathbb{K}$-algebra $\mathbb{V}$, a valuation on $\mathbb{V}$ is a  function $v:\mathsf{Prop}\to\mathbb{V}^+$. As usual, $v$ can be extended to a homomorphism $\llbracket-\rrbracket_v:\mathsf{Str}\to\mathbb{V}^+$. We say that $\mathbb{V},v\models S \Rightarrow T$ if and only if $\llbracket S\rrbracket_v\subseteq \llbracket T\rrbracket_v$.

\begin{theorem}[Completeness]\label{th: wcomp}
Given any sequent $X \Rightarrow Y$ of \textbf{D.NL}$_{\wdia}$, if $\mathbb{V},\upsilon\models X \Rightarrow Y$ for every modal $\mathbb{K}$-algebra $\mathbb{V}$ of finite dimension and any valuation $\upsilon$ on $\mathbb{V}$, then $X \Rightarrow Y$ is a provable sequent in \textbf{D.NL}$_{\wdia}$.
\end{theorem}

As discussed in Section \ref{ssec:displaycalc}, \textbf{D.NL}$_{\wdia}$ is complete and has the finite model property with respect to modal residuated posets.  Therefore, to show Theorem \ref{th: wcomp}, it is enough to show that any finite modal residuated poset can be embedded into the modal residuated lattice of subspaces of a modal $\mathbb{K}$-algebra of finite dimension.

Let $P$ be a finite residuated poset. We will define a modal $\mathbb{K}$-algebra $\mathbb{V}$ and a \textbf{D.NL}$_{\wdia}$-morphism $h:P\to\mathcal{S}(\mathbb{V})$ which is also an order embedding.

Let $n$ be the number of elements of $P$, and let $\{p_1,\ldots,p_n\}$ be an enumeration of $P$.  
Let $\mathbb{V}$ be the $n^2$-dimensional vector space over $\mathbb{K}$ and let $\{e^i_j\mid 1\leq i,j\leq n\}$ be a base. Let $h:P\to\mathbb{V}$ be defined as $$h(p_k)=[e^m_j\mid 1\leq j\leq n\quad \&\quad  p_m\leq p_k].$$ 

We define $\star:\mathbb{V}\times\mathbb{V}\to\mathbb{V}$ on the base as follows: For every $p_k\in\mathcal{P}$ take an surjective map $$\nu_{k}:n\times n\to\{e^m_j\mid 1\leq j\leq n\quad \&\quad  p_m\leq p_k\}$$ such that $\nu_{k}(m,m)=e^k_m$. Define $e^k_m\star e^\ell_r=\nu_{t}(m,r)$, where $p_t=p_k\otimes p_\ell$. This function uniquely extends to a bilinear map and compatible with the scalar product.

We define the relation $R\subseteq\mathbb{V}\times\mathbb{V}$ as follows  $0 R 0$ and  $$\sum_{1\leq i\leq d}\sum_{0\leq j\leq d_i}\alpha^{k^i_j}_{\ell_j} e^{k^i_j}_{\ell_j}R\sum_{1\leq i\leq d}\beta_{j_i}^{m_i} e_{j_i}^{m_i}$$ where $\alpha^{k^i_j}_{\ell_j},\beta_{j_i}^{m_i}\in\mathbb{K}$, $\beta_{j_i}^{m_i}\neq 0$, $p_{k^i_j}\leq \wdia p_{m_i}$ and if $m_i=m_k$ then $j_i\neq j_k$ for $1\leq i,k\leq d$. It is immediate that $R$ satisfies the properties of Definition \ref{definition:R-linear}.


The lemma below shows that $h$ is indeed a \textbf{D.NL}$_{\wdia}$-morphism  which is also an order embedding.
\begin{lemma}
	The following are true for the poset $P$ and $h$ as above.
	\begin{enumerate}
		\item $p\leq q$ if and only if $h(p)\subseteq h(q)$;
		\item $h(p_m\otimes p_k)=h(p_m)\otimes h(p_k)$;
		\item $h(p_m\backslash p_k)=h(p_m)\backslash h(p_k)$;
		\item $h(p_m/p_k)=h(p_m)/h(p_k)$;
		\item $h(\wdia p_k)=\wdia h(p_k)$.
		\item $h(\bbox p_k)=\bbox h(p_k)$.
\end{enumerate}\end{lemma}  
\begin{proof}
	1. Assume that $p\leq q$. Let $\sum_{i,j}\alpha^i_j e^i_j$ an element of $h(p)$ where $p_i\leq p$. Then by assumption $p_i\leq q$, and therefore $\sum_{i,j}\alpha^i_j e^i_j\in h(q)$. For the other direction, assume that $p_m=p\nleq q$, then $e^m_1\notin h(q)$, since each $e^i_j$ is independent from the rest.
	
	2. Let $u\in h(p_m\otimes p_k)$ that is, $u=\sum_{i,j}\alpha^i_j e^i_j$ where $p_i\leq p_m\otimes p_k=p_\ell$. Since $\nu_{\ell}$ is surjective there is $(z^i_j,x^i_j)$ such that $\nu_{\ell}(z^i_j,x^i_j)=e^i_j$. By definition $e^m_{z^i_j}\star e^k_{x^i_j}=e^i_j$. Since $e^m_{z^i_j}\in h(p_m)$ and $e^k_{x^i_j}\in h(p_k)$ for each $i,j$, we have that $$h(p_m)\otimes h(p_k)\ni\sum_{i,j}\alpha^i_j(e^m_{z^i_j}\star e^k_{x^i_j})=\sum_{i,j}\alpha^i_je^i_j=u.$$ 
	
	Conversely let $u\in e(p_m)\otimes e(p_k)$, i.e. $u=\sum_{i,j}\alpha^i_j(e^{m_i}_{m_j}\star e^{k_i}_{k_j})$ where $p_{m_i}\leq p_m$ and $p_{k_i}\leq p_k$. Then $p_{m_i}\otimes p_{k_i}\leq p_m\otimes p_k$. Then, since  $e^{m_i}_{m_j}\star e^{k_i}_{k_j}\in h(p_{m_i}\otimes p_{k_i})$, we have $e^{m_i}_{m_j}\star e^{k_i}_{k_j}\in h(p_m\otimes p_k)$ for each $i$, so $u\in h(p_m\otimes p_k)$.
	
	3. Let $u\in h(p_m\backslash p_k)$. Then $u=\sum_{i,j}\alpha^i_j e^i_j$ where $p_i\leq p_m\backslash p_k$. By adjunction this means that $p_m\otimes p_i\leq p_k$. Pick $\sum_{i',j'}\beta^{i'}_{j'} e^{i'}_{j'}\in h(p_m)$, i.e.\ $p_{i'}\leq p_m$. Notice by monotonicity $p_{i'}\otimes p_i\leq p_k$. Now $$(\sum_{i',j'}\beta^{i'}_{j'} e^{i'}_{j'})\star (\sum_{i,j}\alpha^i_j e^{i}_j)=\sum_{i,i',j,j'}\beta^{i'}_{j'} \alpha^i_j(e^{i'}_{j'}\star e^{i}_{j}).$$ Each of the components are by definition in $h(p_{i'}\otimes p_i)$, and by monotonicity in $h(p_k)$. So for every $w\in h(p_m)$, $w\star u\in h(p_k)$. Therefore $u\in h(p_m)\backslash h(p_k)$.
	
	Conversely, let $u=\sum_{i,j}\alpha^i_j e^{i}_{j}\in h(p_m)\backslash h(p_k)$. Then for every $w\in h(p_m)$, $w\star u\in h(p_k)$. In particular for $w=\sum_{j}e^{m}_{j}$, $$(\sum_{j}e^{m}_{j})\star(\sum_{i,j}\alpha^i_j e^{i}_{j})\in h(p_k)$$. Since $\star$ is bilinear and every element has a unique representation given a base, each  $e^{m}_{j}\star e^{i}_{j}\in h(p_k)$. Let $p_r=p_m\otimes p_i$. By definition of $\nu_{r}$, $e^{m}_{j}\star e^{i}_{j}=e^{r}_{j}\in h(p_k)$ and therefore $p_m\otimes p_i\leq p_k$. That is $p_i\leq p_m\backslash p_k$, i.e. $e^{i}_j\in h(p_m\backslash p_k)$ for each $j$. Therefore $u\in h(p_m\backslash p_k)$.  
	
	4. The proof is the same as item 3.
	
	5. Let $u\in h(\wdia p_k)$,  i.e., $u=\sum_{i}\alpha_{j_i}^{m_i} e_{j_i}^{m_i}$ where $p_{m_i}\leq \wdia p_k$. Since $e^{m_i}_{j_i}Re^k_1$ for each $i$, it follows that  $e^{m_i}_{j_i}\in R^{-1}[h(p_k)]$, for each $i$ and hence $u\in\wdia h(p_k)$. 
	
	Conversely let $u\in \wdia h(p_k)$, i.e. $u\in R^{-1}[h(p_k)]$. By definition of $R$ and the monotonicity of $\wdia$ it follows that $uRe_1^{k}$. So $u=\sum_{i}\alpha_{j_i}^{m_i} e_{j_i}^{m_i}$ where $p_{m_i}\leq \wdia p_k$, i.e.\ $u\in h(\wdia p_k)$.
	
	6.Let $u\in h(\bbox p_k)$. Then $u=\sum_{i}\beta^{m_i}_{j_i} e^{m_i}_{j_i}$ where $p_{m_i}\leq \bbox p_k$. By adjunction this means that $\wdia p_{m_i}\leq p_k$. Let $vRu$ then $v=\sum_{i}\sum_{0\leq j\leq n_i}\alpha^{\ell^i_j}_{r_j} e^{\ell^i_j}_{r_j}$ where $p_{\ell^i_j}\leq \wdia p_{m_i}$. Then $p_{\ell^i_j}\leq p_k$ and therefore $v\in h(p_k)$. Hence $u\in \bbox h(p_k)$.
	
	Conversely, let $u=\sum_{i}\beta^{m_i}_{j_i} e^{m_i}_{j_i}\in \bbox h(p_k)$, i.e.\ $v\in h(p_k)$ for every $v$ such that $vRu$. Notice that $\sum_{i}e^{\ell_i}_{j_i}Ru$ where $p_{\ell_i}=\wdia p_{m_i}$. Since $v\in h(p_k)$ it follows that $\wdia p_{m_i}\leq p_k$ and by adjunction $p_{m_i}\leq\bbox p_k$. Then $e^{m_i}_{j_i}\in h(\bbox p_k)$, for every $i$ and therefore $u\in h(\bbox p_k)$. 
\end{proof}

\begin{remark}
	In the proof above the finiteness of $P$ was used  only to guarantee the dimension of $\mathbb{V}$ to be finite. The same proof holds for an arbitrary modal residuated poset $P$ with a modal $\mathbb{K}$-algebra of dimension $|P\times P|$. That is, every modal residuated poset, and in particular the Lindenbaum-Tarski algebra of \textbf{D.NL}$_{\wdia}$, can be embedded into the lattice of subspaces of some modal $\mathbb{K}$-algebra. 
\end{remark}

\begin{remark}
	In the proof of Theorem \ref{th: wcomp}, we showed that in fact $h$ embeds $P$ into the subalgebra $\{[e_i^j\ \mid\  (i,j)\in S]\mid\  S\subseteq n\times n\}$ which is a Boolean subalgebra of $\mathbb{V}^+$. This is analogous to Buszkowski's proof (see e.g.\ \cite{B-InvNL}) that generalized Lambek calculus is complete with respect to algebraic models based on powerset algebras. 
\end{remark}

\section{Conclusions and further directions}

\paragraph{Our contributions.} In this paper we have taken a duality-theoretic perspective on vector space semantics of the basic modal Lambek calculus and some of its analytic extensions. In a slogan, we have regarded vector spaces (more specifically, modal $\mathbb{K}$-algebras) as Kripke frames. This perspective has allowed to transfer a number of results pertaining to the theory of modal logic to the vector space semantics. Our main contributions are the proof of completeness of the basic modal Lambek calculus $\textbf{D.NL}_{\wdia}$ with respect to the semantics given by the modal $\mathbb{K}$-algebras and a number of ensuing Sahlqvist correspondence results. 
\paragraph{Correspondence and completeness.} In the standard Kripke semantics setting, the completeness of the basic logic and canonicity via correspondence immediately implies that any axiomatic extension of the basic logic with Sahlqvist-type axioms is complete with respect to the elementary class of relational structures defined by the first order correspondents of its axioms. We plan to extend this result to the vector space semantics. 
\paragraph{Adding lattice connectives.} Another direction we plan to pursue consists in extending the present completeness result to the full Lambek calculus signature. Towards this goal, the representation results of \cite{J-ModLatt,J-ReprLatt,F-CMLPS}, which embeds each complemented modular Arguesian lattice into the lattice of subspaces of a vector space (over a division ring), is likely to be particularly relevant. 
\paragraph{Finite vector spaces.} We plan to refine our results so as to give upper bounds on the dimensions of possible witnesses of non derivable sequents. 

\paragraph{Acknowledgements.} We would like to thank Peter Jipsen for numerous observations and suggestions that have substantially improved this paper. We would also like to thank the two anonymous referee for insightful remarks and suggestions.


\end{document}